\documentclass[nolayout]{article} 

\usepackage{hyperref}
\usepackage{url}
\usepackage{amssymb}
\usepackage{dsfont}
\usepackage{amsthm}
\usepackage{amsmath,amsfonts,bm}
\usepackage{authblk,breakcites}
\usepackage{amssymb,amsthm}
\usepackage{geometry}
\usepackage{xcolor}
\usepackage{ifthen}
\usepackage{times}
\usepackage{xargs}
\usepackage{natbib}
\usepackage{graphicx}
\usepackage{algorithm}
\usepackage{algorithmic}

\theoremstyle{plain}
\newtheorem{theorem}{Theorem}[section]
\newtheorem{proposition}[theorem]{Proposition}

\newcommand{\statedim}{d_x}
\newcommand{\obsdim}{d_y}

\newcommand{\vparsp}{\Lambda}
\newcommand{\statesp}{\mathbb{R}^{\statedim}}
\newcommand{\obssp}{\mathbb{R}^{\obsdim}}

\newcommand{\nrmbackwdweight}[3][\parvar]{\bar{w}_{#2-1|#2}^{#1, #3}}
\newcommand{\sample}[2]{\xi_{#1}^{#2}}

\newcommand{\afelbo}[1]{f_{#1}^{\lambda}}
\newcommand{\addfelbo}[2][\parvar]{\tilde{f}_{#2}^{#1}}

\newcommand{\grad}{\nabla}
\newcommand{\gradalt}[1]{\nabla_{#1}}

\newcommand{\parvar}{\lambda}
\newcommand{\md}[1]{g_{#1}}
\newcommand{\hd}[1]{m_{#1}}
\newcommand{\pE}{\mathbb{E}}
\newcommandx\post[2][1=]{
\ifthenelse{\equal{#1}{}}
	{\ensuremath{\phi_{#2}}}
	{\ensuremath{\phi_{#2}^\N}}
}
\newcommand{\rset}{\ensuremath{\mathbb{R}}}
\newcommand{\rmd}{\ensuremath{\mathrm{d}}}
\newcommand{\eqdef}{\ensuremath{:=}}
\newcommand{\eqsp}{\;}
\newcommand{\jointd}[1]{p_{#1}}
\newcommand{\vd}[2][\parvar]{q_{#2}^{#1}} 
\newcommand{\qg}[1]{\ell_{#1}}
\newcommand{\expect}[2]{\mathbb{E}_{#1}\left[#2\right]}
\newcommand{\approxexpect}[2]{\widehat{\mathbb{E}}_{#1}\left[#2\right]}
\newcommand{\ELBO}[2][\parvar]{\mathcal{L}^{#1}_{#2}}
\newcommand{\approxELBO}[2][\parvar]{\widehat{\mathcal{L}}^{#1}_{#2}}

\newcommand{\Hstat}[2][\parvar]{H^{#1}_{#2}}
\newcommand{\Gstat}[2][\parvar]{G^{#1}_{#2}}
\newcommand{\backwd}[1]{b_{#1}}
\newcommand{\fwdpot}[2][\parvar]{\psi_{#2}^{#1}}
\newcommand{\approxHstat}[3][\parvar]{\hat{H}_{#2}^{#1, #3}}
\newcommand{\approxGstat}[3][\parvar]{\hat{G}_{#2}^{#1, #3}}
\newcommand{\paramapproxvarcondE}[2][H]{#1_{\hat{\gamma}_{#2}}^\parvar}
\newcommand{\paramvarcondE}{H_{\gamma}^\parvar}

\title{Importance sampling for online variational learning}

\date{}

\author[$\ddag$]{Mathis Chagneux}
\author[$\star$]{Pierre Gloaguen}
\author[$\dag$]{Sylvain Le Corff}
\author[*]{Jimmy Olsson}

\affil[$\ddag$]{{\small  T\'el\'ecom Paris, Institut Polytechnique de Paris.}}
\affil[*]{{\small  LMBA,Universit\'e Bretagne Sud.}}
\affil[$\dag$]{{\small LPSM, 
       Sorbonne Universit\'e, UMR CNRS 8001.}}
\affil[*]{{\small  Department of Mathematics, KTH Royal Institute of Technology.}}

\begin{document}

\maketitle

\begin{abstract}
 This article addresses online variational estimation in state-space models. 
We focus on learning the smoothing distribution, i.e. the joint distribution of the latent states given the observations, using a variational approach together with Monte Carlo importance sampling. 
We propose an efficient algorithm for computing the gradient of the evidence lower bound (ELBO) in the context of streaming data, where observations arrive sequentially. 
Our contributions include a computationally efficient online ELBO estimator, demonstrated performance in offline and true online settings, and adaptability for computing general expectations under joint smoothing distributions.
\end{abstract}

\section{Introduction}
This article considers online variational estimation problems in state-space models (SSMs), where the observations $Y_{0:t}$\footnote{In all the paper, $a_{u:v}$ is a short-hand notation for $(a_u,\ldots,a_v)$.} depend on a hidden Markov chain denoted by $X_{0:t}$. In this setting, a classical goal is to learn the \textit{smoothing distribution}, which is the conditional distribution of $X_{0:t}$ given $Y_{0:t}$. A popular approach to approximate such a posterior distribution is to rely on particle smoothing, see \citet{chopin2020introduction} and the references therein, but the latter suffers from the curse of dimensionality in the state dimension \cite{bengtsson2008curse}, which has recently motivated the use of variational inference as an alternative. Here, the posterior distribution is approximated by a distribution $\vd{0:t}$, depending on some unknown parameter $\parvar\in\vparsp$. This parameter is then trained by maximizing the evidence lower bound (ELBO): 
\begin{equation}
\label{eq:ELBO}
\ELBO{t} = \pE_{\vd{0:t}}\left[\log \frac{\jointd{0:t}(X_{0:t},Y_{0:t})}{ \vd{0:t}(X_{0:t})}\right] \eqsp,
\end{equation}
where $\jointd{0:t}$ is the joint distribution of the hidden states and the observations. 
Classical optimization procedures rely on the gradient of $\ELBO{t}$ with respect to $\parvar$. Most associated works can be divided into two categories. On the one hand, a large body of methods are only amenable to offline estimation \cite{krishnan2017structured, lin2018variational, johnson2016}, i.e. one needs to have access the entire sequence of observations $Y_{0:t}$ beforehand to compute the ELBO and its gradients. 
On the other hand, some methods provide online procedures via alternative variational objectives which depart from \eqref{eq:ELBO} using additional assumptions on the state dependencies under the variational law \cite{Marino2018AGM,dowling2023real}.
In this paper, we propose an efficient algorithm to update the actual gradient of $\ELBO{t}$ in the context of streaming data, when observations arrive on the fly and are processed only once. 

\begin{itemize}
    \item We propose a computationally efficient online estimator of the ELBO when the variational distribution has a Markov structure. Contrary to computationally intensive Sequential or Markov Chain Monte Carlo methods, our algorithm relies on simple i.i.d samples from the variational distribution.
    \item Experimentally, we demonstrate the performance of our estimators in the offline setting, and assess its performance in online settings with many observations to learn time series representations.
    \item The proposed algorithm is not limited to the online computation and optimization of the ELBO and can be directly adapted to compute more general expectations under distributions admitting a Markovian structure.
\end{itemize}

\section{Model and learning task}

Consider a SSM where the hidden Markov chain in $\statesp$ is denoted by $(X_t)_{t\geqslant 0}$. 
The distribution of $X_0$ has density $\chi$ with respect to the Lebesgue measure $\mu$ and for all $t \geqslant 0$, the conditional distribution of $X_{t+1} $ given $X_{0:t}$ has density $\hd{t}(X_{t},\cdot)$. In SSMs, it is assumed that this state is partially observed through an observation process $(Y_t)_{0 \leqslant t \leqslant T}$ taking on values in $\obssp$. 
The observations $Y_{0:t}$ are assumed to be independent conditionally on $X_{0:t}$ and, for all $0\leqslant t \leqslant T$, the distribution of $Y_t$ given $X_{0:t}$ depends on $X_t$ only and has density $\md{t}(X_t,\cdot)$ with respect to the Lebesgue measure. 
The whole model is then defined by the joint distribution of hidden states and observations:
\begin{equation*}
\label{eq:joint:ssm:dist}
\jointd{0:t}(x_{0:t},y_{0:t}) = \prod_{s=0}^{t}\qg{s}(x_{s-1},x_{s},y_s),
\end{equation*}
where $\qg{s}(x_{s-1},x_s,y_s) \eqdef \hd{s}(x_{s-1}, x_s)\md{s}(x_{s-1},y_{s})$ for $s \geq 1$ and $\ell_0(x_{-1}, x_0, y_0) \eqdef  \chi(x_0)\md{0}(x_{0},y_{0})$.


\subsection{Smoothing in latent data models}
\label{sec:smoothing_latent_data}

A classical learning task in SSMs is \textit{state inference}, which is to estimate the smoothing distribution, \emph{i.e.}, the conditional p.d.f. of $X_{0:t}$ given $Y_{0:t}$.  
This distribution is given by
$$
\post{0:t}(x_{0:t}) \propto \jointd{0:t}(x_{0:t},y_{0:t})\eqsp.
$$
The marginal at time $t$ of this joint distribution is known as the \textit{filtering distribution} at time $t$, and its density w.r.t. the Lebesgue measure is written  $\post{t}$. 
It is straightforward to express the joint smoothing density through the following \textit{backward factorization}:
\begin{equation}
    \label{eq:post:backward:factorization}
    \post{0:t}(x_{0:t}) = \post{t}(x_t)\prod_{s=1}^{t} \backwd{s-1|s}(x_s, x_{s-1}) \eqsp. 
\end{equation}
where, for $1\leqslant s \leqslant t$,
\begin{equation}
    \label{eq:true_backwd}
    \backwd{s-1|s}(x_s, x_{s-1}) = \frac{\hd{s}(x_{s - 1},x_{s})\post{s - 1}(x_{s-1})}{\int \hd{s}(x_{s-1},x_{s})\post{s - 1}(x_{s-1}) \, \rmd x_{s-1}},
\end{equation}
referred to as the \textit{backward kernels}, provide the conditional p.d.f. of $X_{s-1}$ given $(X_{s}, Y_{0:s-1})$.
It is worth noting that \eqref{eq:true_backwd} emphasizes the Markov structure of the smoothing distribution.
Unfortunately, this distribution lacks generally a closed-form expression due to the integral in the denominator and the intractability of the filtering distributions.

\subsection{Backward sequential variational inference}
\label{sec:backwd_seq_inf}
In variational approaches the smoothing distribution $\post{0:t}$ is approximated by choosing a candidate in a parametric family $\{ \vd{0:t}\}_{\parvar \in \vparsp}$, referred to as the \textit{variational} family, where $\vparsp$ is a parameter set. 

A critical point therefore lies in the form of the variational family. 
Motivated by the Markov structure of the smoothing distribution, most works impose structure on the variational family via a factorized decomposition of $\vd{0:t}$ over $x_{0:t}$. 
A variational counterpart of \eqref{eq:post:backward:factorization}, introduced in the work of \citet{campbell2021online}, is to define
\begin{equation}
\label{eq:varpost:backward:factorization}
 \vd{0:t}(x_{0:t})=  \vd{t}(x_t)\prod_{s=1}^{t}\vd{s-1\vert s}(x_{s},x_{s - 1})\eqsp,
\end{equation}
where $\vd{t}$ (resp. $\vd{s - 1\vert s}(x_{s},\cdot)$) are user-chosen p.d.f. whose parameters depend on $Y_{0:t}$ (resp. $Y_{0:s-1}$). 
A key advantage of this factorization is that it respects the true dependencies involved in \eqref{eq:true_backwd}.
Additionally, \citet{chagneux2022amortized} established an upper bound on the error when expectations w.r.t. the smoothing distribution are approximated by expectations w.r.t. variational 
distributions satisfying this backward factorization.
In the following, we consider variational distributions satisfying \eqref{eq:varpost:backward:factorization}.
Variational inference \cite{blei2017} then consists in learning the best $\parvar$ by maximizing the ELBO given in \eqref{eq:ELBO}. 
This is typically done via gradient ascent algorithms that requires to compute the gradient of the ELBO. 
The next sections depict a new algorithm to approximate this gradient recursively. 


\section{Recursion for the gradient of the ELBO}
\label{sec:main:algorithm}
For concise notations, write\footnote{The dependency of each term $\addfelbo{t}$ on $y_t$ is omitted to lighten the notations.}
\begin{equation}
\label{eq:elbo:pair:terms}
\addfelbo{t}(x_{t-1}, x_{t}) = \left\lbrace
\begin{array}{lr}
\log \ell_0(x_{-1}, x_0, y_0) & \text{if } t = 0,\\
\log \frac{\ell_t(x_{t-1}, x_t, y_t)}{q^{\parvar}_{t-1\vert t}(x_t, x_{t-1})} & \text{if }  t > 0
\end{array} \right.
\end{equation}
and  $\afelbo{0:{t}}(x_{0:t}) = \sum_{s = 0}^{t}\addfelbo{s}(x_{s-1}, x_{s})$. Then,
$$
\ELBO{t} = \pE_{ \vd{0:t}}\left[\afelbo{0:{t}}(X_{0:t}) - \log \vd{t}(X_t)\right].
$$
In the following results, all gradients are computed w.r.t. $\parvar$.

\begin{proposition}
\label{prp:elbo:grad:recursion}
The ELBO and its gradient satisfy:
\begin{align}
\ELBO{t} &= \expect{\vd{t}}{\Hstat{t}(X_t)} - \expect{\vd{t}}{\log\vd{t}(X_t)}\\
\grad \ELBO{t}&= \expect{\vd{t}}{\{\grad \log \vd{t} \cdot \Hstat{t}\}(X_t) + \Gstat{t}(X_t)}\label{eq:elbo:grad:update}
,
\end{align}
where $\Hstat{t}(x_t)$ is a function from $\statesp$ to $\mathbb{R}$ satisfying the recursion
\begin{equation}
\label{eq:Hstat:recursion}
    \Hstat{t}(x_t) = \pE_{q^{\parvar}_{t-1\vert t}(x_t, \cdot)}\left[\Hstat{t-1}(X_{t-1}) + \addfelbo{t}(X_{t-1}, x_t)\right]\eqsp,
\end{equation}
with $\Hstat{0}(x_0) = \addfelbo{0}(x_{-1}, x_{0})$, and $\Gstat{t}(x_t)$  is a function from $\statesp$ to $\vparsp$ satisfying $\Gstat{0}(x_0) = 0$ and
\begin{equation}
\label{eq:Gstat:recursion}
\Gstat{t}(x_t) = \mathbb{E}_{\vd{t-1|t}(x_t,\cdot)}\left[\Gstat{t-1}\left(X_{t-1}\right) + \grad \log \vd{t-1|t}(x_t, X_{t-1}) \mathsf{H}_{t-1}^\parvar(X_{t-1}, x_t)\right]\eqsp,
\end{equation}
with  $\mathsf{H}_{t-1}^\parvar(x_{t-1}, x_t) = \Hstat{t-1}(x_{t-1}) + \addfelbo{t}(x_{t-1}, x_t)$. 
\end{proposition}

\begin{proof}
The proof is postponed to Appendix~\ref{appdx:proof:prp:elbo:recursion}.
\end{proof}

For a fixed sequence of length $t$, recursive computation of $\grad \ELBO{t}$ therefore consists in (i) computing recursively $\grad \Hstat{t}(x_t)$ from 0 to $t$ using the recursions of Proposition \ref{prp:elbo:grad:recursion} and (ii) computing the final expectation \eqref{eq:elbo:grad:update}. 
As this final expectation is w.r.t. the chosen variational distribution $\vd{t}$, standard Monte Carlo sampling can be used to approximate it.
It remains to approximate online the intermediate functions $\Hstat{t}(x_t)$ and $G_t^\parvar(x_t)$, involving conditional expectations with respect to the kernels $\vd{t-1\vert t}$. 
Note that these kernels aim to approximate the true backward kernels $b_{t-1|t}$, which  depend only on observations $Y_{0:t-1}$, \emph{i.e.}, only on the past, and hence, are prone to online learning.

\section{Approximation of the gradient}
\label{sec:main_algorithm:mc} 

\subsection{Offline computation}
\label{sec:offline:computation}

For the sake of clarity, we first depict the algorithm for a fixed size sequence of observations $Y_{0:T}$. For a fixed $\parvar$, we use Proposition \ref{prp:elbo:grad:recursion} to compute $\grad \ELBO{T}$. 
We suppose we have access to a sequence of variational distributions $(\vd{t})_{t\geq 0}$ and $(\vd{t - 1\vert t})_{t\geq 1}$ which can be evaluated and sampled from easily. 
The choice of such sequences is discussed in Section \ref{sec:computational_considerations}.

\paragraph{Initialization. } Simulate a $N$-sample $\lbrace \sample{0}{j} \rbrace_{1\leq j \leq N} \overset{i.i.d}{\sim} \vd{0}$, set $\approxHstat{0}{j} = \Hstat{0}(\sample{0}{j}), \approxGstat{0}{j} = \Gstat{0}(\sample{0}{j})=0$. 
The key point for the propagation step is that these two functionals are known (at $t = 0$) only on a \textit{finite support}.

\paragraph{Recursive approximation of $\Hstat{t}$ and $\Gstat{t}$. } At time $t$, simulate a $N$-sample $\lbrace \sample{t}{i} \rbrace_{1\leq i \leq N} \overset{i.i.d}{\sim} \vd{t}$. 
$\Hstat{t}(\sample{t}{i})$ and  $\Gstat{t}(\sample{t}{i})$ are approximated respectively by:
\begin{align}
\approxHstat{t}{i} =& \sum_{j=1}^N \nrmbackwdweight{t}{i,j}
\left(\approxHstat{t-1}{j} + \addfelbo{t}(\sample{t-1}{j}, \sample{t}{i}) \right) \eqsp, \label{eq:approx:Hstat:offline}\\ 
\approxGstat{t}{i} =& \sum_{i=1}^N \nrmbackwdweight{t}{i,j}\left\{\approxGstat{t - 1}{j} + \grad \log \vd{t-1|t}(\xi_{t}^i, \xi_{t-1}^j) \left( \approxHstat{t-1}{j} + \addfelbo{t}(\xi_{t-1}^j, \xi_t^i) \right)\right\}\eqsp,  \label{eq:approx:Gstat:offline}
\end{align}
where
\begin{equation}
    \label{eq:backwardweights:offline}
    \nrmbackwdweight{t}{i,j} = \frac{\vd{t-1 \vert t}(\sample{t}{i}, \sample{t-1}{j})/\vd{t-1}(\sample{t-1}{j})}{\sum_{k=1}^N \vd{t-1 \vert t}(\sample{t}{i}, \sample{t-1}{k})/\vd{t-1}(\sample{t-1}{k})}\eqsp.
\end{equation}
Estimators \eqref{eq:approx:Hstat:offline} and \eqref{eq:approx:Gstat:offline} are \textit{self-normalized importance sampling} (SNIS) estimators of equations \eqref{eq:Hstat:recursion} and \eqref{eq:Gstat:recursion}.  
Equation \eqref{eq:backwardweights:offline} gives the shared importance weights of these estimator. 
It is worth noting that we cannot do direct Monte Carlo approximations of \eqref{eq:Hstat:recursion}-\eqref{eq:Gstat:recursion} by simulating samples from $\vd{t-1\vert t}(\sample{t}{i},\cdot)$, as the functionals $\Hstat{t-1}$ and $\Gstat{t-1}$ would have no approximation on such samples. 
The use of importance sampling is thus mandatory to update the approximations. 
It is known, though, that performance of such estimator strongly rely on the importance distribution. 
Section \ref{sec:computational_considerations} proposes an efficient implementation to link the proposal distribution $\vd{t}$ to the target $\vd{t - 1\vert t}$.
The self-normalization in \eqref{eq:backwardweights:offline} is motivated by computational considerations as it reduces the variance of the estimator in our simulation settings.

\paragraph{Estimators at final time. } At final time $T$, simulate a $N$-sample $\lbrace \sample{T}{i} \rbrace_{1\leq i \leq N} \overset{i.i.d}{\sim} \vd{T}$, compute $\approxHstat{T}{i}$ and $\approxGstat{T}{i}$ using \eqref{eq:approx:Hstat:offline}-\eqref{eq:approx:Gstat:offline}, and approximate the ELBO and its gradient with:
\begin{align}
\approxELBO{T} &= \frac{1}{N}\sum_{i = 1}^N\left\{ \approxHstat{T}{i} - \log\vd{T}(\sample{T}{i})\right\}, \label{eq:approx:ELBO:T}\\
\widehat{\grad} \ELBO{T} &= \frac{1}{N}\sum_{i = 1}^{N}\left\{ \grad \log \vd{T}(\sample{T}{i}) \cdot \approxHstat{T}{i} + \approxGstat{T}{i}\right\}\eqsp. \label{eq:approx:grad:ELBO:T}
\end{align}

Note that an appealing alternative to estimator \eqref{eq:approx:grad:ELBO:T}  would be to perform auto-differentiation of estimator \eqref{eq:approx:ELBO:T}, as it would avoid to perform the recursion \eqref{eq:approx:Gstat:offline}.
However, in our setting, this approach is flawed because the approximation \eqref{eq:approx:Hstat:offline} is biased due the SNIS. 
Building gradients via auto-differentiation of a biased estimator can lead to catastrophic divergence, especially in the case of ELBO maximization which is based on an upper bound. 
Typically, the auto-differentiation will lead to the parameters that maximize the bias of the approximation.

\subsection{Online computation}
\label{sec:online:computation}

In the online computation framework, we aim at approximating the gradient at each time $t$, using the new observation, and updating the current parameter $\parvar_t$ using this gradient. 

\paragraph{Initialization. } Starting from an initial guess $\parvar_0$, the initialization is the same as in Section \ref{sec:offline:computation}. 
Then the first gradient is approximated with:
$$
\widehat{\grad} \ELBO[\parvar_0]{0} = \frac{1}{N} \sum_{i = 1}^{N} \grad \log \vd[\parvar_0]{0}(\sample{0}{i}) \cdot \approxHstat[\parvar_0]{0}{i} + \approxGstat[\parvar_0]{0}{i}\eqsp.
$$
Then, $\parvar_1$ is set by updating $\parvar_0$ using this gradient, typically by setting $\parvar_1 = \parvar_0 + \gamma_0\widehat{\grad} \ELBO[\parvar_0]{0}$ for some step size $\gamma_0$.

\paragraph{Recursive approximation of $\Hstat[\parvar_t]{t}$ and $\Gstat[\parvar_t]{t}$. }

At time $t$, simulate a $N$-sample $\lbrace \sample{t}{i} \rbrace_{1\leq i \leq N} \overset{i.i.d}{\sim} \vd[\parvar_t]{t}$. 
$\Hstat[\parvar_t]{t}(\sample{t}{i})$ and  $\Gstat[\parvar_t]{t}(\sample{t}{i})$ are approximated respectively by:
\begin{align}
\approxHstat[\parvar_t]{t}{i} =& \sum_{j=1}^N \nrmbackwdweight[\parvar_t]{t}{i,j}
\left(\approxHstat[\parvar_{t-1}]{t-1}{j} + \addfelbo[\parvar_t]{t}(\sample{t-1}{j}, \sample{t}{i}) \right) \eqsp, \label{eq:approx:Hstat:online}\\ 
\approxGstat[\parvar_t]{t}{i} =& \sum_{i=1}^N \nrmbackwdweight[\parvar_t]{t}{i,j}\left\{\approxGstat[\parvar_{t-1}]{t - 1}{j} + \grad \log \vd[\parvar_{t}]{t-1|t}(\xi_{t}^i, \xi_{t-1}^j)  \left( \approxHstat[\parvar_{t - 1}]{t-1}{j} + \addfelbo[\parvar_t]{t}(\xi_{t-1}^j, \xi_t^i) \right)\right\}\eqsp,  \label{eq:approx:Gstat:online}
\end{align}
where
\begin{equation}
    \label{eq:backwardweights:online}
    \nrmbackwdweight[\parvar_t]{t}{i,j} = \frac{\vd[\parvar_{t}]{t-1 \vert t}(\sample{t}{i}, \sample{t-1}{j})/\vd[\parvar_{t - 1}]{t-1}(\sample{t-1}{j})}{\sum_{k=1}^N \vd[\parvar_{t}]{t-1 \vert t}(\sample{t}{i}, \sample{t-1}{k})/\vd[\parvar_{t - 1}]{t-1}(\sample{t-1}{k})}\eqsp.
\end{equation}
Equations \eqref{eq:approx:Hstat:online}-\eqref{eq:backwardweights:online} are almost identical to \eqref{eq:approx:Hstat:offline}-\eqref{eq:backwardweights:offline}.
The key difference is presence of quantities, on the right hand sides of \eqref{eq:approx:Hstat:online}-\eqref{eq:approx:Gstat:online}, that were computed with $\parvar_{t - 1}$, thus introducing new approximations. 
These approximations, which are commonly made in recursive maximum likelihood settings, are mandatory for practical implementation of online learning. 
Section \ref{sec:experiments} shows that these approximation, which would make theoretical study more complex, still lead to good results in practice.

\section{Computational considerations}
\label{sec:computational_considerations}

\paragraph{Defining backward kernels using forward potentials. }
Equations \eqref{eq:backwardweights:offline} and \eqref{eq:backwardweights:online} suggest that the performance of the proposed algorithm would strongly rely on the definition of variational distributions, and on the link between $\vd[\parvar_{t-1}]{t-1}$ and $\vd[\parvar_t]{t-1\vert t}$.
We therefore introduce additional structure in the variational family given by \eqref{eq:varpost:backward:factorization}, using potential functions $\fwdpot{t}:\statesp \times \statesp \rightarrow \rset$ to link these distributions. Specifically, we prescribe that, for all $t \geq 1$, 
\begin{equation}
\label{eq:def:varbackwd}
\vd[\parvar_t]{t-1|t}(x_t, x_{t-1})  \propto \vd[\parvar_{t-1}]{t-1}(x_{t-1})\fwdpot[\parvar_t]{t}(x_{t-1}, x_t)\eqsp.
\end{equation}
The functions $\fwdpot[\parvar_t]{t}$ can be made arbitrarily complex, such that, for all $t$, the backward variational kernel $\vd[\parvar_t]{t-1|t}$ has arbitrarily complex dependencies w.r.t $x_t$. 

\paragraph{Backward sampling. }

Computing the backward weights of  \eqref{eq:backwardweights:online} has the drawback of a $O(N^2)$ complexity due to the computation normalizing constant, which can be prohibitive when $N$ is large (typically for high dimensional state spaces).
One solution, introduced by \citet{olsson2017efficient} in the context of particle smoothing, is the \textit{backward sampling}. 
At $t$, given $\sample{t}{i}$, sample independently $M$ indexes $j_1,\dots, j_M \in \{1,\ldots,N\}$ from the multinomial distribution with weights $\{\nrmbackwdweight{t}{i,j}\}_{j \leq N}$, and replace \eqref{eq:approx:Hstat:online} by $\sum_{k=1}^M (\approxHstat[\parvar_{t-1}]{t-1}{j_k} + \addfelbo[\parvar_t]{t}(\sample{t-1}{j_k}, \sample{t}{i}))/M$.
\citet{olsson2017efficient} showed that $M$ can be much smaller than $N$ (typically, $M=2$), and then, this can lead to great improvement in complexity.
Here, noting that $\nrmbackwdweight{t}{i,j} \propto_{j} \fwdpot{t}(\sample{t-1}{j}, \sample{t}{i})\eqsp,$
the backward sampling step is done with an accept-reject mechanism without having to compute the normalizing constant of the weights. 
We refer the reader to  \citet{olsson2017efficient}, \citet{gloaguen2022pseudo} and \citet{dau2022complexity} for alternative backward sampling approach. 

\paragraph{Parameterization of variational distributions. }
\label{sec:implementation}
In practice, employing the backward factorization under decomposition \eqref{eq:def:varbackwd} in the online setting requires explicit parameterization, for all $t \geq 0$, of a distribution $\vd{t}$ and a potential $\fwdpot{t}$ (not necessarily normalized) both of which depend on observations up to time $t$ at most. 
For computational efficiency, each $\vd{t}$ is chosen as a p.d.f. from a parametric distribution within the exponential family. 
This family is denoted as $\mathsf{P} = \{P_\eta\}_{\eta \in \mathcal{E}}$ where $\eta$ is the corresponding natural parameter and $\mathcal{E}$ the parameter space for this family (typically be the family of Gaussian distributions defined on $\statesp$, in our experiments). Let $\eta_t^\parvar$ be  the parameter of $\vd{t}$. 
To ensure that distributions $\vd{t - 1\vert t}$ belong to the same family, we impose that 
$$
    \fwdpot{t}(x_{t-1},x_t) = \exp{(\tilde{\eta}_{t}^{\parvar}(x_t) \cdot T(x_{t-1}))}\eqsp,
$$
where $\tilde{\eta}_{t}^\parvar(x_t) = \mathsf{MLP}^\parvar(x_t)$\footnote{$\mathsf{MLP}$ is used to denote a multi-layer perceptron.} and $T(x_{t-1})$ are a natural parameter and a sufficient statistic for the family $\mathsf{P}$. 
Then, thanks to \eqref{eq:def:varbackwd}, $\vd{t-1|t}(x_t, \cdot)$ will be a p.d.f.  from $\mathsf{P}$ with natural parameter $\eta_{t-1|t}^\parvar = \eta_{t-1}^\parvar + \tilde{\eta}_t^\parvar$. 
In this convenient setting, the backward kernels $\vd{t-1|t}$ can have arbitrarily complex dependencies on $x_t$ while their p.d.f is analytically derived from the potentials. 
This eliminates the necessity to calculate normalizing constants (required, for instance, when computing \eqref{eq:approx:Gstat:online}), all the while avoiding the reduction of these kernels to mere transformations or linearizations (e.g., linear-Gaussian kernels). 
For the parameters of $\vd{t}$ two main approaches exist.
\begin{itemize}
    \item \textit{Amortized schemes} where the parameters of $\vd{t}$ are updated using a parameterized mapping at every time $t$. This can be done using intermediate quantities $a_t \in \mathsf{A}$ (where $\mathsf{A}$ is a user-defined space), such that $a_t = \mathsf{MLP}^\parvar(a_{t-1}, y_t)$, and $\eta_t^\parvar = \mathsf{MLP}^{\parvar}(a_t)$. Initialization is performed using a random parameter $a_{-1}$, which may be fixed or learnt. 
    Amortized schemes are computationally efficient (as knowledge from previous predictions is used to produce the current parameters), but require manually defining complex mappings. 
    The recursions may be analytical (and not rely on a MLP), for example when $\vd{0:t}$ is the smoothing distribution of a linear-Gaussian, or when conjugacy is further leveraged to update the parameters $(\eta_t^\parvar)_{t \geq 0}$ (see Appendix \ref{appdx:impl_details}).
    Whatever the case, while the number of parameters becomes independent of $t$, the computational burden of the backpropagation through the states $(a_s)_{s \leq t}$ grows linearly with $t$. 
    To prevent this, a solution is to truncate backpropagation, i.e. to assume that $(a_s^\parvar)_{s \leq t-\Delta}$ is independent of $\parvar$ for some $\Delta$. 
    \item \textit{Non-amortized schemes} where each $\vd{t}$ and $\fwdpot{t}$ have their own parameter $\eta_t$ and $\tilde{\eta}_t$, not related to those at time $t-1$. In this case, the optimized vector $\parvar$ contains the parameters $(\eta_t)_{t \geq 0}$, and the number of parameters then grows linearly with $t$. This scheme modifies equation \eqref{eq:Gstat:recursion} (see  Appendix \ref{appdx:details:non:amortized} for details).
\end{itemize}

\paragraph{Variance reduction of the gradient estimator. }

Equations \eqref{eq:elbo:grad:update} and \eqref{eq:Gstat:recursion} involve computing expectations of \textit{score functions}, i.e. expectations of the form $\expect{q^\parvar}{\grad \log q^\parvar(X) \cdot f(X)},$ for some p.d.f. $q^\parvar(X)$. 
As studied in \citet{mohamed2020}, direct Monte Carlo estimation of the score-function yields high variance and should typically not be used without a proper variance reduction technique. 
The most straightforward approach is to design a \textit{control variate}. 
Exploiting the fact that $\expect{q^\parvar}{\grad \log q^\parvar(X)} = 0$, the target expectation is then equal to $\mathbb{E}_{q^\parvar}[\grad \log q^\parvar(X) (f(X) - \expect{q^\parvar}{f(X)})]$.
Fortunately, Monte Carlo estimates of $\mathbb{E}_{\vd{t}}[\Hstat{t}]$ and  $\mathbb{E}_{\vd{t - 1\vert t}}[
\Hstat{t-1}(X_{t-1}) + \addfelbo{t}(X_{t-1}, x_t)]$ are directly given as byproducts of algorithm of Section \ref{sec:online:computation}, namely by $N^{-1}\sum_{i = 1}^N \approxHstat{t}{i}$ and $\{\approxHstat{t}{i}\}_{1\leq i \leq N}$. 
Our methodology comes built-in with variance reduction without the need to recompute additional quantities.
As an alternative to this variance reduction technique, a natural consideration arises regarding the potential use of the reparametrization trick, as it often results in Monte Carlo estimators with lower variance compared to those obtained using the score function.
However, implementing the reparametrization trick in this context necessitates expressing $\grad \ELBO{t}$ as an expectation with respect to a random variable $Z_{0:t}$ that does not depend on $\parvar$. Moreover, the recursive expression of this expectation at time $t + 1$ must be derivable from its predecessor, which poses a non-trivial challenge.
For example, in the classical case where $q_{0:t}$ is the p.d.f. of a multivariate Gaussian random variable with mean $\mu$ and variance $\Sigma$, and the expectation is taken w.r.t. $Z_{0:t} \sim \mathcal{N}(0, I_{\statedim \times (t + 1)})$, such a recursion is not feasible as the ELBO is no longer an additive functional when $X_{0:t}$ is replaced by  $\mu + \Sigma^{\frac{1}{2}}Z_{0:t}$.

Appendix~\ref{appdx:full:algo} provides the full algorithm using the backward sampling approach and the control variate estimate.

\section{Related work}

\paragraph{Sequential Monte Carlo smoothing. }

The presented methodology draws from recent advances in sequential Monte Carlo (SMC methods) for SSMs, especially by i) proposing Monte Carlo approach for the approximation of conditional expectations under the backward kernels, and ii)  introducing the structure \eqref{eq:def:varbackwd} to link $\vd[\parvar_{t-1}]{t-1}$ and $\vd[\parvar_{t}]{t-1\vert t}$.
As a major difference, though, we emphasize that here, all Monte Carlo samples are obtained by i.i.d. sampling, avoiding the \textit{selection/mutation} steps of SMC, which lead to dependant samples over time. 
We refer the reader to \citet{douc2014nonlinear} (Section 11) for a presentation of the general concepts underpinning this class of algorithms, and \citet{olsson2017efficient}, \citet{gloaguen2022pseudo} and \citet{dau2022complexity} for recent works on this class of algorithms. 
It is worth noting here that the algorithm proposed here can actually be implemented for every expectations of \textit{any additive functional} as defined in these references, and not only the ELBO or it's gradient. 

\paragraph{Sequential variational inference. }

In sequential variational inference, early works tackling smoothing focus on offline scenarios with a forward factorization different from the one used here \cite{johnson2016, krishnan2017structured, Lin2018VariationalMP, Marino2018AGM}. 
A drawback of forward factorizations is the incompatibility with online setting.
\citet{campbell2021online} provides the first work which explores online variational additive smoothing for recursive computation of the ELBO and its gradients.  
The main difference with our approach is the way to approximate conditional expectations with respect to backward kernels. 
The authors rely on functional approximations of the conditional expectations at each timestep. 
The major drawback of this solution is that it requires running an inner optimization (to learn the best regression function as a proxy of the target conditional expectation) \textit{at every iteration $t$}, which may be very costly, and requires an additional choice of the regression functions. 

Recently, \citet{chagneux2022amortized} established the first theoretical result on error control in sequential variational inference, building upon the backward factorization proposed in \citet{campbell2021online}.
This result, coupled with the potential for online learning, serves as motivation for our algorithm.

A distinct line of research \cite{Marino2018AGM, zhao2020variational, dowling2023real} chooses to trade smoothing for filtering by targeting the marginal distributions $(\post{t})_{t \leq 0}$ at each timestep with variational distributions $\vd{t}$ that depend only on the observations up to $t$. 
A distinctive trait of these works is the additional assumption that, at $t$, $\vd{t}$ is a good approximation of $\post{t}$ which can be used in further timesteps to learn the next approximations. 
In practice, this can hardly be verified especially under Gaussian variational families. 
That said, some of the ideas proposed in these works may be relevant for the smoothing problem, which we discuss more extensively as a perspective in Section \ref{sec:perspectives}. 

\section{Experiments}
\label{sec:experiments}

\subsection{Linear-Gaussian HMM}

We first evaluate our algorithm on a Linear-Gaussian HMM, which admits analytical smoothing distribution. such that optimal smoothing is available. We set:
\begin{align*}
X_0 \sim \mathcal{N}(\mu_0, Q_0),~X_t &= A X_{t-1} + \nu_t\eqsp, t \geq 1\eqsp,\\
Y_t &= B X_t + \epsilon_t\eqsp, t \geq 0\eqsp,
\end{align*}
where $\nu_t$ and $\epsilon_t$ are Gaussian centered noises with unknown variances $Q$ and $R$, and  $\mu_0, Q_0, A$ and $B$ are unknown parameters with appropriate dimensions.
In this case the Kalman smoothing recursions yield the smoothing distribution, which is a Gaussian distribution. 
It is then possible to choose a variational model parameterized by $\parvar$ which gets arbitrarily close to the true posterior by prescribing that each $\vd{t}$ is the p.d.f. of a Gaussian distribution and $\vd{t - 1\vert t}$ is a Gaussian kernel with linear dependance on $x_t$. 
In this case, the ELBO can also be computed recursively in closed-form because the conditional expectations $(\Hstat{t})_{t \geq 0}$ are quadratic forms.

\paragraph{Learning in an offline setting. }

We first evaluate our algorithm on a sequence of fixed-length $T$ to showcase that the proposed framework indeed enables to perform a gradient ascent algorithm. 
As an oracle baseline, we can compute the closed-form ELBO and its associated gradient via the reparameterization trick. 
As an alternative, an unbiased offline Monte Carlo estimate of the ELBO is obtained by drawing trajectoires using the backward dynamics given by the kernels $(\vd{t-1|t})_{t \leq T}$ and using the reparameterization trick to obtain is gradient. We refer to this method as \textit{backward MC}. To compare the methods at hand, we evaluate our ability to perform gradient-ascent to optimize the ELBO with respect to $\parvar$. For our recursive method, we choose $\Delta=2$ to truncate the backpropagation, as we observe that $\Delta < 2$ prevents our method from converging altogether, while $\Delta > 2$ only improves convergence speed by a small margin. The experiment is run using $10$ different parameters for the generative model, $\statedim=\obsdim=10$, $T=500$ and $N=2$ for the two methods involving Monte Carlo sampling.
Figure \ref{fig:training_curve_lgm} displays the evolution of the ELBO  using both approaches. 
It shows the convergence of our score-based solution to the correct optimum given by the analytical computations. 
This is particularly appealing and notably demonstrates that our online gradient-estimation method may perform well using few samples. 
In practice, we observe that the variance reduction introduced in Section \ref{sec:computational_considerations} is crucial in reaching such performance. 
Finally, despite the added cost of updating the intermediate quantities for the gradients at each timestep, we observe that the computational times of our solution is about 2.5 time slower than the optimization based on the oracle gradient computed analytically (in average, 89.5 ms per gradient step for our method, 37.1 for the oracle method).
We want to emphasize here that we do not advocate for our method in the context of offline learning (hence for time series of small length).

\begin{figure}[t]
    \centering
    \includegraphics[width = .7\textwidth, trim={0.5cm 0.5cm 0.5cm 0.5cm},clip]{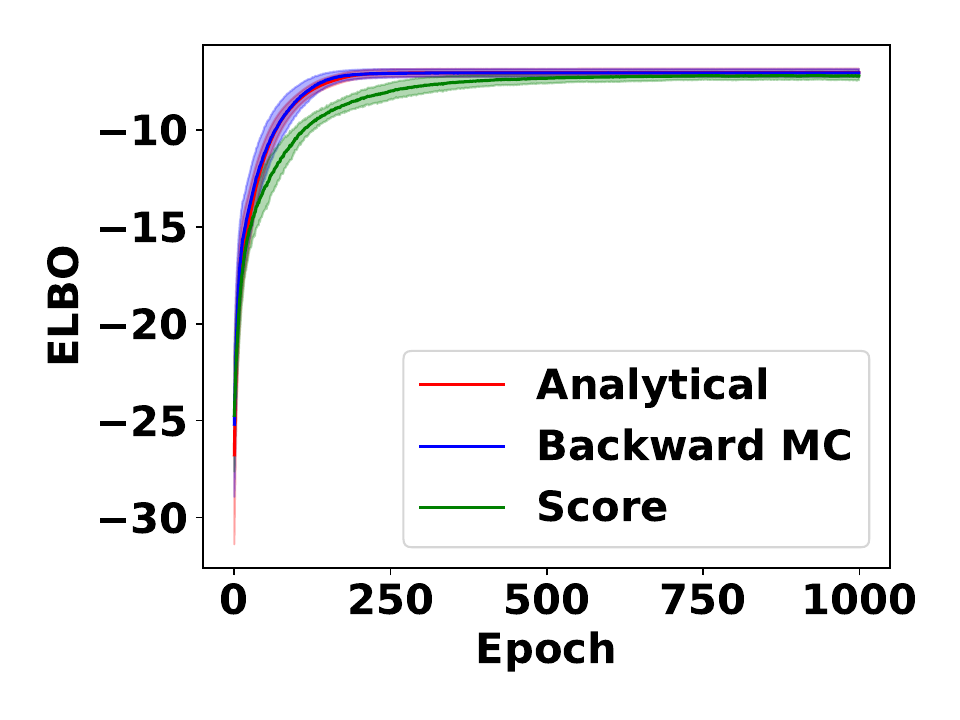}
    \caption{Evolution of $\ELBO{T}/T$ computed with three different methods and with three different types of gradients estimates. Full lines: means of the 10 replicates. Shaded lines: standard deviations of the 10 replicates.}
    \label{fig:training_curve_lgm}
\end{figure}

\paragraph{Online learning from streaming data. }
In a second setting, we keep the same generative and variational models and dimensionality but generate a large sequence of $T=500000$ observations and simulate the optimization of the joint ELBO $\ELBO{T}$. 
The purpose here is to update the variational parameters online, i.e. by discarding already seen data at each step.
In the context of stochastic optimization, since $\ELBO{T} =  \sum_{t=0}^T \ELBO{t} - \ELBO{t-1}$ (with the convention $\ELBO{-1} = 0$), the right quantity to optimize becomes $\grad \{\ELBO{t} - \ELBO{t-1}\}$. 
In practice we update $\parvar_{t +1}$ by setting:
\begin{equation}
    \label{eq:online_updates}
    \parvar_{t+1} = \parvar_t + \gamma_{t+1}\left(\grad \ELBO[\parvar_t]{t} - \grad \ELBO[\parvar_{t-1}]{t-1}\right)\eqsp,
\end{equation}
in order to avoid recomputing the previous gradient\footnote{This approximation is typically made in traditional recursive maximum likelihood methods}. In Figure \ref{fig:training_100k}, we plot the evolution of the ELBO through this optimization process. 

\begin{figure}[t]
    \centering
    \includegraphics[width=0.7\textwidth]{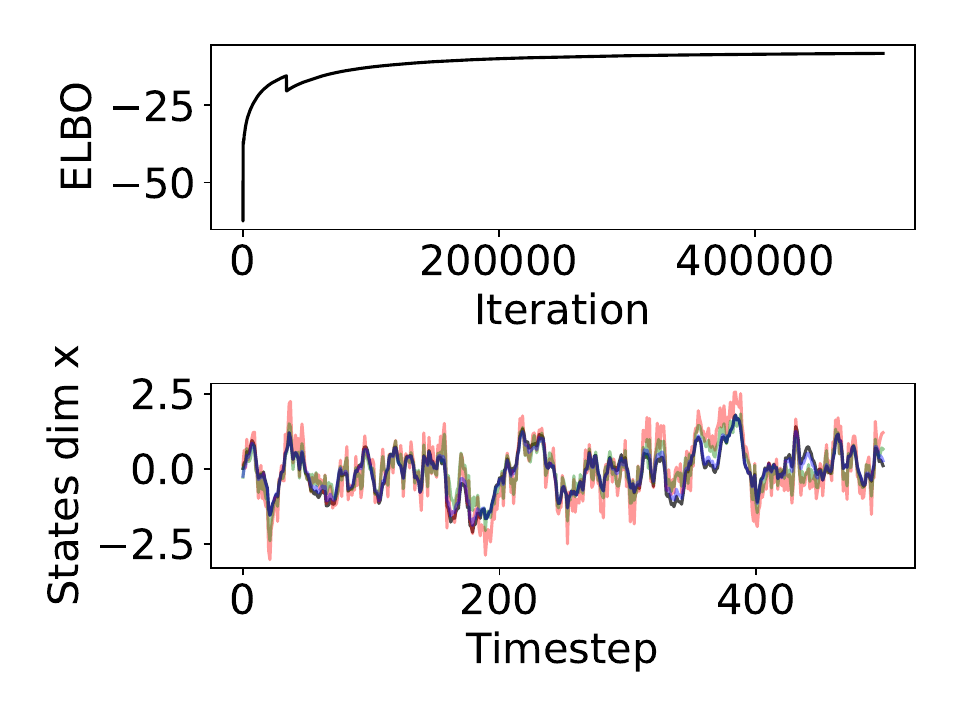}
    \caption{Top: evolution of $\frac{1}{t}\approxELBO{t}$ during the online learning. Bottom: evaluation on a test sequence of 500 observations from the same generative model (over 1 particular dimension, in black) for parameters obtained at iterations $1$ (red), $10,000$ (green) and $500,000$ (blue).}
    \label{fig:training_lgm_100k}
\end{figure}

\subsection{Chaotic recurrent neural network.}
\label{sec:xp:chaotic:rnn}
We now consider the model used in \citet{campbell2021online}, where $X_0 \sim \mathcal{N}(0, Q)$, and,  for $t\geq 1$:
\begin{align*}
X_t =& X_{t-1} + \frac{\Delta}{\tau}\left(\gamma W \tanh{(X_{t-1})} - X_{t-1}\right) + \eta_t\eqsp,\\
Y_t =& X_t + \epsilon_t,~t\geq 0 \eqsp,
\end{align*}
where $\eta \overset{\text{i.i.d.}}{\sim} \mathcal{N}(0, Q)$ is an isotropic Gaussian distribution and $\epsilon$ is a Student-$t$ distribution, these two distributions being mutually independent. 
The hyperparameters chosen are the ones of \citet{campbell2021online} (see Appendix~\ref{appdx:chaotic:rnn}). 

\paragraph{Learning in an offline setting. }
Again, we start by evaluating the performance of our gradients against the backward trajectory sampling approach run on the same model, for a sequence of fixed length $T=500$ with $\statedim=\obsdim=5$. 
For the variational family, we choose the setting presented in the section \ref{sec:implementation} where a linear-Gaussian kernels with parameters $(A^\parvar, Q^\parvar)$ is used both to update the distributions $(\vd{t})_{t \geq 0}$ and the backward kernels.
We run gradient-ascent on $\parvar$ by performing gradient steps using the quantity $\grad \ELBO{T}/T$ approximated via backward trajectory sampling and via our score-based method. 
As before, we use the same hyperparameters and optimization schemes for both methods. 
Table \ref{table:smoothing_rmse_chaotic_dim_5} reports the performance against the true states, averaged over dimensions, i.e. the quantity 
$$\Delta_{T}^\parvar=\frac{1}{T}\sum_{t=1}^T \sqrt{\frac{1}{\statedim}\sum_{k=1}^{\statedim} \left(x_t^{*(k)} - \expect{\vd{0:T}}{X_t^{(k)}}\right)^2}\eqsp.$$

\begin{table}[t]
    \small
    \centering
    \begin{tabular}{||c||c|c||} 
        \hline
        Gradients & $\Delta_{T,p}^\parvar$ ($\times 10^{-2}$) & Avg. time \\ [0.5ex] 
        \hline\hline
        Score-based & 13.5 $\pm$ 0.7 ({\bf 12.2}) & 173 ms \\ 
        \hline
        Backward sampling & 11.9 $\pm$ 0.4 ({\bf 11.4}) & 17 ms \\
        \hline
    \end{tabular}
    \caption{RMSE between the true states $x_t^*$ and the predicted marginal means $\expect{\vd{0:T}}{X_t}$ and average time per gradient step.}
\label{table:smoothing_rmse_chaotic_dim_5}
\end{table}

\paragraph{Recursive gradients in the offline setting. }

Even when we have access to an entire sequence of observations $y_{0:T}$, it can still be beneficial to use the recursive gradients approach for faster convergence. 
Indeed, when gradients are only available after processing the whole \textit{batch}\footnote{i.e. the whole set of observation}, the best we can do at optimization given a fixed number of observations $T$ is to update the parameter with
\begin{equation}
\parvar^{(k+1)}= \parvar^{(k)} + \gamma_{k+1}\grad \ELBO[\parvar^{(k)}]{T}\eqsp,
\label{eq:epoch_update}
\end{equation}
where one such update is usually referred to as an "epoch", and $\parvar^{(k)}$ is the value of estimated parameter after $k$ epochs. 
Using the recursive gradients, one may perform $T$ intermediate updates within an epoch using
\begin{equation}
\parvar_{t+1}^{(k)} = \parvar_{t}^{(k)} + \gamma_{t+1}^{(k)}\left\{\grad \ELBO[\parvar_t^{(k)}]{t+1} - \grad \ELBO[\parvar_{t-1}^{(k)}]{t}\right\}\eqsp,
\label{eq:online_inside_epoch_update}
\end{equation}
and 
$$\parvar_{0}^{(k+1)} = \parvar_{T}^{(k)}\eqsp,$$
i.e. inside one epoch we optimize $\parvar$ recursively on the observations. 
We compare the two options by optimizing on 10 different sequences of $T=500$ observations, performing $10$ epochs on each, using updates of the form \eqref{eq:epoch_update} for the backward trajectory sampling approach and using updates of the form \eqref{eq:online_inside_epoch_update} with our score-based approach. 
Figure \ref{fig:training_curve_online_inside_epoch}, displays the epoch-wise training curves for each method with $\statedim = \obsdim=5$, where we observe that optimizing with intermediate updates of  \eqref{eq:online_inside_epoch_update} converges faster overall.

\begin{figure}
     \centering
    \includegraphics[width = .7\textwidth, , trim={0.5cm 0.5cm 0.5cm 0.5cm},clip]{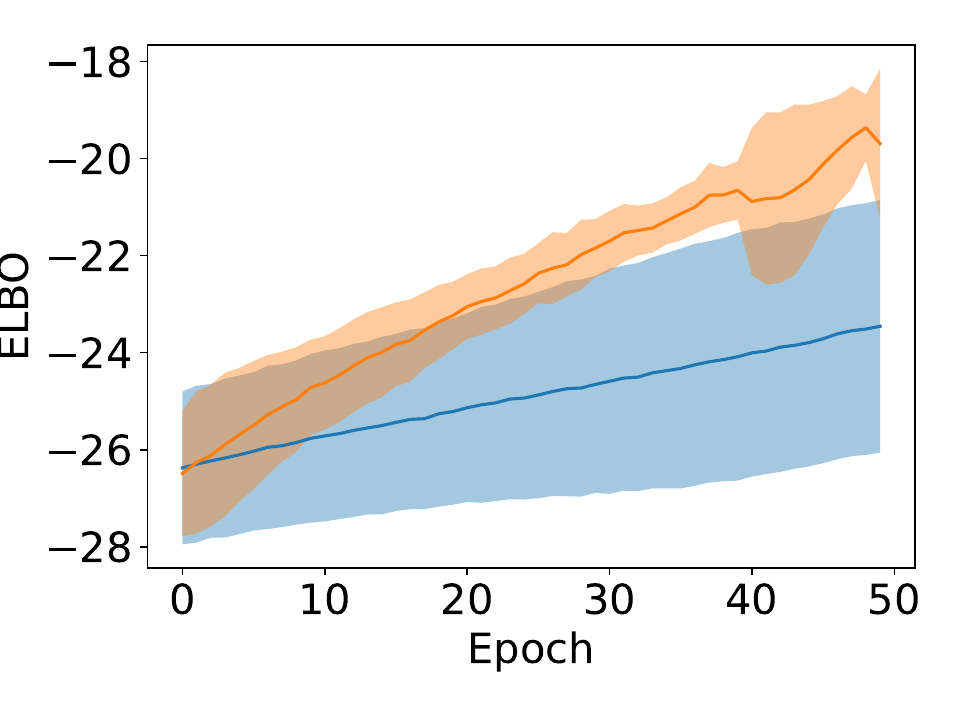}
     \caption{Evolution of $\ELBO{T}/T$ for $\parvar = \parvar^{(k)}$ when performing with temporal updates inside an epoch (via recursive gradients) or without (as in Figure \ref{fig:training_curve_lgm}), $k \in \{0,\ldots,50\}$. The full lines are the average over the 10 different runs, the shaded lines are the standard deviations across these runs.}
    \label{fig:training_curve_online_inside_epoch}
\end{figure}

\paragraph{Comparison with \citet{campbell2021online}.}

The proposed method of this paper mainly differs from \citet{campbell2021online} in the way we approximate the backward statistics $H_t^\parvar$ by using a recursive sampling approach rather than a regression approach. 
In order to compare the two approaches, we reproduce the experiment of appendix B.2 of \citet{campbell2021online}, where the authors evaluate their ability to predict the hidden state one step backward (therefore performing 1-step smoothing).

Specifically, we aim at evaluating the quality of our approach to estimate the conditional law of $X_{t-1}$ given $Y_{0:t}$ and of $X_{t}$ given $Y_{0:t}$ by evaluating $\expect{\vd{t-1:t}}{X_{t-1}}$ and $\expect{\vd{t}}{X_t}$, where $\parvar$ is learnt using the same setting as \citet{campbell2021online}, i.e. a non-amortized scheme  (see \textit{Non-amortized schemes}, Section~\ref{sec:computational_considerations} and Appendix ~\ref{appdx:details:non:amortized}).
The details of implementation are given in Appendix~\ref{appdx:chaotic:rnn}

Table \ref{table:1_step_smoothing_chaotic_rnn}, reports the average 1-step smoothing errors and filtering errors, i.e. the quantities 
$$
\kappa^{(1)}_{T}=\frac{1}{T-1} \sum_{t=1}^{T-1} \sqrt{\frac{1}{\statedim}\sum_{k=1}^{\statedim} \left(\expect{\vd{t-1:t}}{X_{t-1}^{(k)}} - x_{t-1}^{*(k)}\right)^2}
$$ 
and 
$$\kappa^{(2)}_{T}=\frac{1}{T}\sum_{t=1}^T \sqrt{\frac{1}{\statedim}\sum_{k=1}^{\statedim} \left(\expect{\vd{t}}{X_t^{(k)}} - x_t^{*(k)}\right)^2}
$$ 
when training our method in these two settings with $\statedim=\obsdim=5$. 
We also report the errors the computational times for the two methods averaged over 8 runs using 8 different generative models (hence 8 different sequences). 

One can see that for comparable results, our approach based on Monte Carlo for estimating the backward expectation is about 5 times faster than the regression approach.

\begin{table}[t]
    \centering
    \small
    \begin{tabular}{||c||c|c|c||} 
        \hline
        Method & Smooth. & Filt. & Time\\ [0.5ex] 
        \hline
        Ours & 8.9 (0.2) & 10.3 (0.2) & 1 ms \\ 
        \hline
        \citet{campbell2021online} & 9.2 (0.2) & 10.3 (0.2) & 4.8 ms \\
        \hline
    \end{tabular}
    \caption{RMSE ($\times 10^{-2}$) between the true states $x_t^*$ and  $\approxexpect{\vd{t-1:t}}{X_{t-1}}$ (column Smooth.) and  $\approxexpect{\vd{t}}{X_t}$ (Filt.) (with the standard error), and average time per gradient step.}
    \label{table:1_step_smoothing_chaotic_rnn}
\end{table}

\paragraph{Online learning from streaming data. }

Finally, we evaluate the performance in the true online setting when training on a sequence of $T=100,000$ observations using parameter updates of the form of \eqref{eq:online_updates}. 
We choose $\statedim = \obsdim = 10$ and $N=100$ particles. 
To parameterize $\vd{0:t}$ we use the amortized model presented at the beginning of this section. 
Table \ref{table:chaotic_streaming_data} provides the smoothing and filtering RMSE against the true states at the end of optimization. 
We also show the inference performance on new sequences generated under the same generative model. 
The results clearly highlight that the fitted $\parvar$ is relevant for new sequences, and illustrates the performance of our method in the amortized setting. 
This scheme is then appealing when one wants to train a single model on a long stream of incoming data, then re-use it for offline state inference on new sequences of arbitrary length.

\begin{table}[t]
    \centering
    \small
    \begin{tabular}{||c||c|c||} 
        \hline
        Sequence & Smoothing RMSE & Filtering RMSE  \\ [0.5ex] 
        \hline
        Training & 0.281 & 0.311 \\ 
        \hline
        Eval & 0.278 ($\pm$ 0.01) & 0.305 ($\pm$ 0.014) \\
        \hline
    \end{tabular}
    \caption{Smoothing and filtering RMSE values for the training sequence and other sequences drawn from the same generative model, when $\parvar$ is learnt online.}
    \label{table:chaotic_streaming_data}
\end{table}

\begin{figure}[t]
    \centering
    \includegraphics[width=0.7\textwidth, trim={0.5cm 0.5cm 0.5cm 0.5cm},clip]{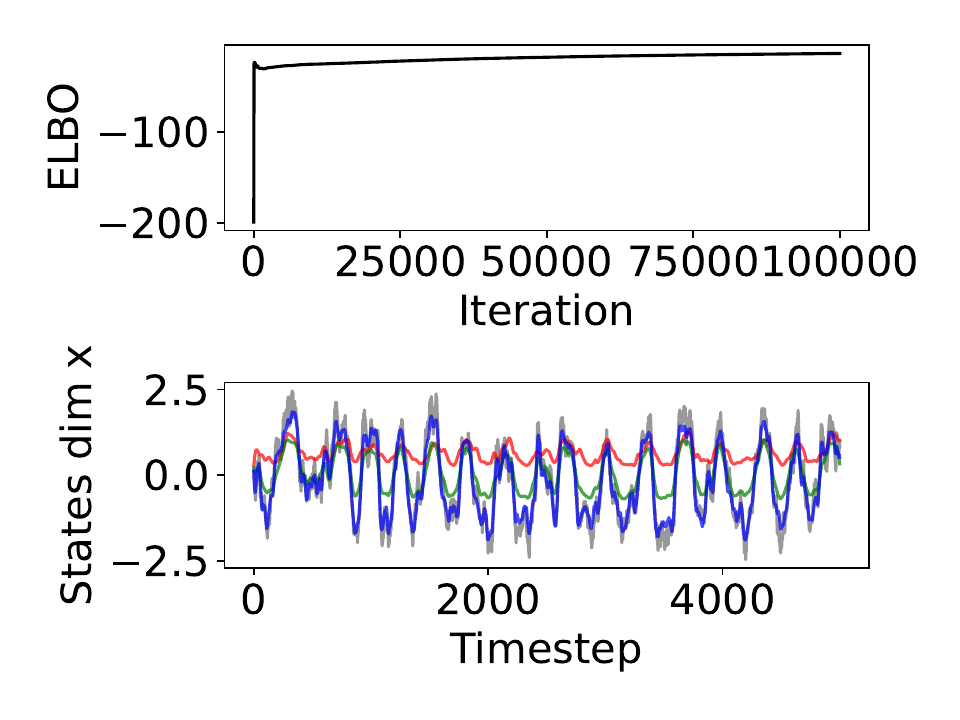}
    \caption{Top: evolution of $\frac{1}{t}\approxELBO{t}$ during the online learning. Bottom: evaluation on a test sequence of 5,000 observations from the same generative model (over 1 particular dimension, in grey) for parameters obtained at iterations $1$ (red), $10,000$ (green) and $100,000$ (blue).}
    \label{fig:training_100k}
\end{figure}

\section{Discussion}
\label{sec:perspectives}
For future research, we identify two directions that could benefit from further investigation. 
First, we have only implemented the versions of our algorithm that rely on exponentially conjugated potentials, and as such more general parameterizations need to be evaluated. 
In practice, when the forward potentials $(\fwdpot{t})_{t \geq 0}$ are arbitrarily parameterized functions, it is expected that more flexible joint variational approximations can be obtained, and hence better results under complex nonlinear models.

Then, a more thorough analysis could be conducted to study the proper \textit{stepwise} objective to optimize in situations where parameter updates are performed at every timestep. 
In this work, we have relied on the decomposition $\mathcal{L}_t^\parvar = \sum_{s=1}^t \mathcal{L}_s^\parvar - \mathcal{L}_{s-1}^\parvar$ as a justification to solve the optimization problem in $\parvar$ via online stochastic gradient updates which maximize the ELBO over time. In contrast, works like \citet{zhao2020variational, dowling2023real} develop an online variational optimization procedure by deriving lower bounds on the incremental likelihood.
In practice, these solutions depart from the original ELBO and formulate intermediate optimization problems at each timestep by deriving a "single step" ELBO from the Kullback-Leibler divergence between $\post{t}$ and $\vd{t}$ at each $t \geq 0$. 
As such, they do not target the smoothing distributions, and a joint variational distribution on the state sequence $X_{0:t}$ is only available in mean-field form $\vd{0:t} = \prod_{s=0}^t \vd{s}$, which does not capture dependencies between the states. 
In \citet{dowling2023real}, however, the focus is put on learning the true model transitions, and they introduce a "hybrid" version of the \textit{predictive} distribution of $X_t$ given $Y_{0:t-1}$ defined as $\bar{q}_t^{\parvar}(x_t) = \mathbb{E}_{\vd{t-1}}[\hd{t}(X_{t-1}, x_t)]$ for all $t \geq 0$, which is used to propagate the variational distributions. In the context of smoothing, we may similarly design variational backward kernels which rely on the true dynamics, i.e. for all $x_t \in \statesp$,
$q_{t-1|t}^{\parvar}(x_t, x_{t-1}) \propto \vd{t-1}(x_{t-1})\hd{t}(x_{t-1}, x_t)$, 
in which case the normalizing constant is precisely $\bar{q}_t^{\parvar}(x_t)$ as defined above.

\bibliography{biblio}
\bibliographystyle{apalike}

\clearpage

\appendix

\section{Proof of proposition \ref{prp:elbo:grad:recursion}}
\label{appdx:proof:prp:elbo:recursion}

We start by the definition of the ELBO:
\begin{align*}
    \ELBO{t} &= \expect{\vd{0:t}}{\log \frac{\jointd{0:t}(X_{0:t},Y_{0:t})}{ \vd{0:t}(X_{0:t})}}&
    \\
    &=  \expect{\vd{0:t}}{\log \frac{\prod_{s = 0}^t\qg{s}(X_{s-1}, X_s, Y_s)}{\vd{t}(X_t)\prod_{s = 1}^t\vd{s-1\vert s}(X_{s}, X_{s - 1})}} & \text{by \eqref{eq:joint:ssm:dist} and \eqref{eq:varpost:backward:factorization}}
    \\
    &= \expect{\vd{0:t}}{\sum_{s = 0}^t \log \frac{\qg{s}(X_{s-1}, X_s, Y_s)}{\vd{s-1\vert s}(X_{s}, X_{s - 1})} - \log\vd{t}(X_t)} & \text{Posing } \vd{-1\vert 0}(x_{0}, x_{-1}) = 1
    \\
    &= \expect{\vd{0:t}}{\sum_{s = 0}^t \addfelbo{s}(X_{s-1}, X_s) - \log\vd{t}(X_t)} & \text{By the defintion \eqref{eq:elbo:pair:terms}}
    \\
    &= \expect{\vd{0:t}}{\sum_{s = 0}^t \addfelbo{s}(X_{s-1}, X_s)} - \expect{\vd{t}}{\log\vd{t}(X_t)} &
    \\
    &= \int \left\{\sum_{s = 0}^t \addfelbo{s}(x_{s-1}, x_s)\right\}\vd{t}(x_t) \prod_{s = 1}^{t}\vd{s-1\vert s}(x_{s}, x_{s-1}) \rmd x_1\dots \rmd x_t  - \expect{\vd{t}}{\log\vd{t}(X_t)}&
    \\
    &= \expect{\vd{t}}{\Hstat{t}(X_t)} - \expect{\vd{t}}{\log\vd{t}(X_t)}\eqsp,&
\end{align*}
where 
\begin{align*}
\Hstat{t}(x_t) \eqdef 
\int \left\{\sum_{s = 0}^t \addfelbo{s}(x_{s-1}, x_s)\right\} \prod_{s = 1}^{t}\vd{s-1\vert s}(x_{s}, x_{s-1}) \rmd x_1\dots \rmd x_{t-1} = \expect{\vd{0:(t-1)\vert t}}{\afelbo{0:t}(X_{0:t-1}, x_t)}\eqsp,
\end{align*}
with $\afelbo{0:t}$  the function defined in Section \ref{sec:main:algorithm}, and $\vd{0:(t-1)\vert t}(x_{0:t-1}, x_t) = \prod_{s = 1}^t  \vd{s-1\vert s}(x_{s}, x_{s-1})$ is a p.d.f., for every $x_t$, for $X_{0:t-1}$. Then, notice that:
\begin{align*}
\Hstat{t}(x_t) &= \int\left(\int \left\{\sum_{s = 0}^{t-1} \addfelbo{s}(x_{s-1}, x_s) + \addfelbo{t}(x_{t-1}, x_t) \right\} \prod_{s = 1}^{t-1}\vd{s-1\vert s}(x_{s}, x_{s-1}) \rmd x_1\dots \rmd x_{t-2} \right)\vd{t-1\vert t}(x_{t}, x_{t-1}) \rmd x_{t-1}\\
 &= \int \left(\Hstat{t -1}(x_{t-1}) + \addfelbo{t}(x_{t-1}, x_t)\right)\vd{t-1\vert t}(x_{t}, x_{t-1}) \rmd x_{t-1}\\
 &= \expect{\vd{t-1\vert t}}{\Hstat{t -1}(X_{t-1}) + \addfelbo{t}(X_{t-1}, X_t)}\eqsp.
\end{align*}
This development then shows the given expression of the ELBO and the recursion over $\Hstat{t}$. Consider the gradient of the ELBO, with respect to $\parvar$.
\begin{align*}
    \grad\ELBO{t} &= \grad \expect{\vd{t}}{\Hstat{t}(X_t)} - \grad \expect{\vd{t}}{\log\vd{t}(X_t)}\eqsp\\
    &= \grad\int \left(\Hstat{t}(x_t) - \log \vd{t}(x_t)\right)\vd{t}(x_t)\rmd x_t\\
    &= \int \left(\grad\Hstat{t}(x_t) - \grad\log \vd{t}(x_t)\right)\vd{t}(x_t)\rmd x_t + \int \left(\Hstat{t}(x_t) - \log \vd{t}(x_t)\right) \grad \vd{t}(x_t) \rmd x_t \\
    &= \expect{\vd{t}}{\grad\Hstat{t}(X_t)} - \overset{=0}{\expect{\vd{t}}{\grad\log\vd{t}(X_t)}} + \int \left(\Hstat{t}(x_t) - \log \vd{t}(x_t)\right) \grad \left(\log \vd{t}(x_t)\right) \vd{t}(x_t)  \rmd x_t\\
    &= \expect{\vd{t}}{\grad\Hstat{t}(X_t) + \left(\Hstat{t}(x_t) - \log \vd{t}(x_t)\right)\times  \grad \log \vd{t}(x_t)}, 
\end{align*}
where $\expect{\vd{t}}{\grad\log\vd{t}(X_t)}=0$ is justified by the fact that $\expect{\vd{t}}{\grad\log\vd{t}(X_t)} = \int \grad \vd{t}(x_t)\rmd x_t = 0$. 

Now, if we denote $\Gstat{t}(x_t) = \grad \Hstat{t}(x_t)$.
\begin{align*}
    \Gstat{t}(x_t) &= \grad \expect{\vd{0:(t-1)\vert t}}{\afelbo{0:t}(X_{0:t-1}, x_t)}\\
    &=\expect{\vd{0:(t-1)\vert t}}{\left\{\grad \log \vd{0:(t-1)\vert t} \times \afelbo{0:t}\right\}(X_{0:t-1}, x_t)}+ \expect{\vd{0:(t-1)\vert t}}{\grad \afelbo{0:t}(X_{0:t-1}, x_t)}.
\end{align*}
Here, it turns out that\footnote{This would not be the case for another additive functional than the one of the ELBO.}, by definition of $\afelbo{0:t}$ :
$$
\grad \afelbo{0:t}(X_{0:t-1}, x_t) = -\grad \log \vd{0:(t-1)\vert t}(X_{0:t-1}, x_t),
$$
therefore, the expectation of this term is equal to 0, and:
$$
\Gstat{t}(x_t) = \expect{\vd{0:(t-1)\vert t}}{\left\{\grad \log \vd{0:(t-1)\vert t} \times \afelbo{0:t}\right\}(X_{0:t-1}, x_t)}\eqsp.
$$
Finally, it remains to show the wanted recursion for $\Gstat{t}(x_t)$:
\begin{align}
    \Gstat{t}(x_t) =& \expect{\vd{0:(t-1)\vert t}}{\left(\grad \log \vd{0:(t-2)\vert t-1}(X_{0:t-1}) + \grad \log \vd{t - 1\vert t}(X_{t-1}, x_{t}) \right)\left(\afelbo{0:t-1}(X_{0:t-1}) + \addfelbo{t}(X_{t-1}, x_t)\right)} \label{eq:appdx:beg:G}\\
    =& \expect{\vd{t-1\vert t}}{\Gstat{t-1}(X_{t-1})} \label{eq:appdx:Gtm1} \\
    &+  \expect{\vd{t-1\vert t}}{\grad \log \vd{t - 1\vert t}(X_{t-1}, x_{t})\left(\expect{\vd{0:(t-2)\vert t-1}}{\afelbo{0:t-1}(X_{0:t-1})} + \addfelbo{t}(X_{t-1}, x_t)\right)} \label{eq:appdx:G:term:Htm1}\\
    &+\expect{\vd{t-1\vert t}}{\addfelbo{t}(X_{t-1}, x_t)\times \expect{\vd{0:(t-2)\vert t-1}}{\grad \log \vd{0:(t-2)\vert t-1}(X_{0:t-1})}}.\label{eq:appdx:end:G}
\end{align}
On the inner expectation of \eqref{eq:appdx:G:term:Htm1} we recognize $\Hstat{t - 1}$ and \eqref{eq:appdx:end:G} is again equal to 0.
We therefore have the wanted result:
$$
\Gstat{t}(x_t) = \expect{\vd{t-1\vert t}}{\Gstat{t-1}(X_{t-1}) + \grad \log \vd{t - 1\vert t}(X_{t-1}, x_{t})\times \left(\Hstat{t-1}(X_{t-1}) + \addfelbo{t}(X_{t-1}, x_t) \right)}\eqsp.
$$

\section{Details on the non amortized scheme}
\label{appdx:details:non:amortized}
In the non amortized scheme, $\parvar$ is a set of disjoint parameters, each of them corresponding to a specific time step.
Namely $\parvar = \lbrace \parvar^0, \dots, \parvar^t\rbrace$ . 
In the notations of the article, the estimate $\parvar_{t-1}$ of $\parvar$ after having processed observations $y_{0:{t-1}}$ is an estimate of the set $\lbrace\parvar^0, \dots, \parvar^{t-1}\rbrace$. 
Therefore, the gradient of the ELBO  will only be with respect to $\parvar^t$. 
This affects the expression of the statistic $\Gstat[\parvar{t}]{t}$, and one can see in equations \eqref{eq:appdx:beg:G}-\eqref{eq:appdx:end:G} that the term \eqref{eq:appdx:Gtm1} will now be 0 when the gradient is taken w.r.t. $\parvar^t$. 
This means that this term no longer has to be propagated.
Indeed, as we set $\vd[\parvar^t]{t-1|t}(x_t, x_{t-1}) \propto \vd[\parvar_{t-1}]{t - 1}(x_{t-1})\fwdpot[\parvar^t]{t}(x_{t-1}, x_t)$, the gradient of the ELBO w.r.t. $\parvar^t$ will be
\begin{align*}
\gradalt{\parvar^t} \ELBO[\parvar^t]{t} =& \mathbb{E}_{\vd[\parvar^t]{t}}
\left[ 
\mathbb{E}_{\vd[\parvar^t]{t - 1\vert t}}
\left[
\grad \log \vd{t}(X_t) \times \addfelbo[\parvar^t]{t}(X_{t-1}, X_t) \right.
\right.
\\
&+ 
\left.\left.\gradalt{\parvar^t} \log \vd[\parvar^t]{t - 1\vert t}(X_{t-1}, x_{t})\times \left(\Hstat[\parvar_{t-1}]{t-1}(X_{t-1}) + \addfelbo[\parvar^t]{t}(X_{t-1}, X_t) \right)\right]\right]\eqsp.
\end{align*}

This gradient will be estimated using Monte Carlo in the same way as in the algorithm. In \citet{campbell2021online}, the inner conditional expectation is estimated with a regression approach of Appendix \ref{appdx:func_regression} instead of importance sampling as in our approach.

\section{Using exponential conjugacy to process observations}
\label{appdx:impl_details}

To further reduce the computational cost, one may actually leverage exponential conjugacy both to update the parameters of the distributions $(\vd{t})_{t \geq 0}$ and to derive the parameters of the backward kernels $(\vd{t-1|t})_{t \geq 0}$. This is possible, for example, whenever $\mathsf{P}$ is the Gaussian family. Denoting $\vec{\eta}$ the function such that $\vec{\eta}^\parvar(x_{t-1})$ is the natural parameter vector of the linear-Gaussian kernel $\mathcal{N}(A^\parvar x_{t-1}, Q^\parvar)$, and $\vec{\phi}^\parvar(x_{t-1}, x_t)$ the density of that latter kernel evaluated at $x_t$, then closed form updates can be derived for all parameters at any timestep when choosing

\begin{itemize}
    \item $\eta_t^\parvar = \expect{\vd{t-1}}{\vec{\eta}^\parvar(X_{t-1})} + \bar{\eta}_{y_t}^\parvar$ where $\bar{\eta}_{y_t}^\parvar = \mathsf{MLP}^\parvar(y_t)$ is a natural parameter. Here, the expectation on the right hand side is analogous to the predict step in Kalman filtering, but assimilation of the observation can involve a complex nonlinear mapping, as originally proposed in \citet{johnson2016}.
    \item $\fwdpot{t}(x_{t-1}, x_t) \propto \vec{\phi}^\parvar(x_{t-1}, x_t)$, in which case $\fwdpot{t}(\cdot, x_t)$ is still conjugated to $\vd{t-1}$ for any $x_t$, and the parameters of $\vd{t-1|t}$ can be derived as explained above simply by deriving the natural parameter which makes $\vec{\phi}^\parvar(\cdot, x_t)$ conjugated to $\vd{t-1}$.
\end{itemize}

In this setting the backward kernels are linear and Gaussian, and the only neural network involved in the variational approximation is used to assimilate the observations. 
Additionally, the parameters $(A^\parvar, Q^\parvar)$ are shared between the updates for $(\vd{t})_{t \geq 0}$ and those for $(\vd{t-1|t})_{t \geq 1}$, which is analogous to the true model recursions where the forward transition kernels are involved both in the filtering recursions and in the definition of the backward kernels.

\section{Functional regression}
\label{appdx:func_regression}
Here, we recall the alternate option used in \citet{campbell2021online} to propagate approximation of the backward expectations.
Denoting $\mathcal{F} = \left\lbrace g: \rset^p \to \rset^{\statedim}, \mathbb{E}_{\vd{t}}[\|g(X_t)\|_2] < \infty\right\rbrace$,  $\Hstat{t}(x)$ satisfies (by definition of conditional expectation):
$$
\Hstat{t} = \operatorname*{argmin}_{g \in \mathcal{F}} \mathbb{E}_{\vd{t-1:t}(X_{t-1}, X_t)} \|g(X_t) -  [\Hstat{t-1}(X_{t-1})  + \addfelbo{t}(X_{t-1}, X_t)]\|_2,
$$
which provides a regressive objective for learning an approximation of $\Hstat{t}$.
In practice authors restrict the minimization problem to a subset of $\mathcal{F}$, a parametric family of functions (typically, a neural network) parameterized by $\gamma$, belonging to $\Gamma \subset \rset^{d_\gamma}$, and learn this by approximating the expectation with Monte Carlo method. 
Namely, the authors propose to estimate $\Hstat{t}$ by $\paramapproxvarcondE{t}$ where 
\begin{equation}
\label{eq:campbell:minimization}
\hat{\gamma}_t = \operatorname*{argmin}_{\gamma \in \Gamma}
\frac{1}{N}\sum_{k = 1}^N \|
\paramvarcondE(\xi^k_t) - 
[\paramapproxvarcondE{t-1}(\xi^k_{t-1})  + \tilde{h}_t(\xi^k_{t-1}, \xi^k_t)]\|_2\eqsp,
\end{equation}
where $\left\lbrace (\xi_{t-1}^i, \xi_{t}^i) \right\rbrace_{i=1,\dots,N}$ is  an i.i.d. sample under the variational joint distribution of $(X_{t-1}, X_t)$ which has density $\vd{t-1:t} = \vd{t}\vd{t-1|t}$.
Upon convergence, $\paramapproxvarcondE{t}$ is then used in the successive recursions (in the type of $\eqref{eq:approx:Hstat:online}$). 

\section{Experiments settings}

\subsection{Appendix for section \ref{sec:xp:chaotic:rnn}}
\label{appdx:chaotic:rnn}

\paragraph{Parameters for the chaotic RNN}
We choose the same hyperparameters than \citet{campbell2021online} with $\Delta=0.001$, $\tau=0.025$, $\gamma=2.5$, $2$ degrees of freedom and a scale of $0.1$ for the Student-$t$ distribution, and define $Q = \text{diag}(0.01)$. 

\paragraph{Implementation settings for the comparison with \citet{campbell2021online}}

Each $\parvar^t$ contains the parameter $\eta_t = (\mu_t, \Sigma_t)$ of the distribution $\vd[\parvar^t]{t} \sim \mathcal{N}(\mu_t, \Sigma_t)$ and the parameter $\tilde{\eta}_t$ of the function $\fwdpot[\parvar^t]{t}$. 
For this latter function, we match the number of parameters of \citet{campbell2021online} by defining $\fwdpot[\parvar^t]{t}(x,y) = \exp{(\tilde{\eta}_t(y) \cdot T(x))}$ with $\tilde{\eta}_t(y) = (\tilde{\eta}_{t,1}(y), \tilde{\eta}_{t,2})$ where $y \mapsto \tilde{\eta}_{t,1}(y)$ is a multi-layer perceptron with 100 neurons from $\statesp$ to $\statesp$, and $\tilde{\eta}_{t,2}$ is a negative definite matrix. We follow the optimization schedules of \citet{campbell2021online} with $K=500$ gradient steps at each timestep.

\section{Full algorithm with backward sampling and control variate}
\label{appdx:full:algo}
\begin{algorithm}[ht]
\caption{One iteration of the online gradient ascent algorithm (for $t\geq 1$) in the amortized scheme}
\begin{algorithmic}
\REQUIRE \textit{}
\begin{itemize}
    \item Previous statistics $\{\approxGstat[\parvar_{t-1}]{t-1}{i}, \approxHstat[\parvar_{t-1}]{t-1}{i}\}_{i=1}^N$;
    \item Previous samples $\{\sample{t-1}{i}\}_{i=1}^N$;
    \item Intermediate quantity $a_{t-1}$ (see Section \ref{sec:computational_considerations} \textit{Parameterization of variational distributions});
    \item Current parameter estimate $\parvar_{t}$; 
    \item Step size $\gamma_t$ for the gradient optimization procedure;
    \item New observation $y_t$.
\end{itemize}
\ENSURE $\{\approxGstat[\parvar_{t}]{t}{i}, \approxHstat[\parvar_{t}]{t}{i}\}_{i=1}^N$, $\parvar_{t + 1}$, $a_t$.
\STATE
Compute $a_t = \mathsf{MLP}^{\parvar_t}(a_{t-1}, y_t)$.
\STATE 
Compute $\eta^{\parvar}_t = \mathsf{MLP}^{\parvar_t}(a_{t})$, the parameters of $q_t^{\parvar_{t}}$ and sample $\{\sample{t}{i}\}_{i=1}^N$ i.i.d. with distribution $\vd[\parvar_{t}]{t}$.
\FOR{$i=1$ to $i=N$}
\STATE Compute $\tilde{\eta}_{t}^{\parvar_{t},i} = \mathsf{MLP}^{\parvar_{t}}(\sample{t}{i})$
\FOR{$j=1$ to $j=M$}
\STATE  \textit{// Backward sampling step, M is the number of backward samples} 
\STATE Sample $(j_k)_{1\leq k \leq M} \overset{\text{i.i.d.}}{\sim} \mathsf{Cat}(\{\nrmbackwdweight[\parvar_t]{t}{i,j}\}_{1\leq j \leq N})$ with the weights of \eqref{eq:backwardweights:online}. 
\ENDFOR
\STATE Compute \textit{// Recall that each term $\addfelbo[\parvar_t]{t}$ depends on $y_t$.}
\begin{align*}
\approxHstat[\parvar_{t}]{t}{i} &= 
\frac{1}{M}\sum_{k=1}^M
\left\{
\approxHstat[\parvar_{t - 1}]{t - 1}{j_k} + \addfelbo[\parvar_{t}]{t}(\sample{t-1}{j_k},\sample{t}{i})\right\}\eqsp,\\
\approxGstat[\parvar_{t}]{t}{i} &= 
\frac{1}{M}\sum_{k=1}^M
\left\{
\approxGstat[\parvar_{t - 1}]{t-1}{j_k} + \grad \log \vd[\parvar_{t}]{t-1|t} (\sample{t-1}{j_k},\sample{t}{i})\left( \approxHstat[\parvar_{t-1}]{t}{j_k} + \addfelbo[\parvar_{t}]{t}(\sample{t-1}{j_k},\sample{t}{i}) - \approxHstat[\parvar_{t}]{t}{i} \right)\right\}\eqsp.
\end{align*}
\STATE  \textit{// Note the difference with \eqref{eq:approx:Gstat:online} and the inclusion of control variate $\approxHstat{t}{i}$ for the computation of $\approxGstat{t}{i}$} 
\STATE \textit{// }$\grad\vd[\parvar_{t}]{t-1|t} (\sample{t-1}{j_k},\sample{t}{i})$ \textit{is typically computed with automatic differentiation}
\ENDFOR
$$
\parvar_{t + 1} = \parvar_{t} + \frac{\gamma_t}{N}\sum_{i=1}^N
\left\{\approxGstat[\parvar_{t}]{t}{i} + \grad \log \vd[\parvar_{t-1}]{t}(\sample{t}{i})\left(\approxHstat[\parvar_{t}]{t}{i} - \frac{1}{N}\sum_{k=1}^N \approxHstat[\parvar_{t}]{t}{k}\right)\right\}\,.
$$
\STATE \textit{// Note the difference with \eqref{eq:approx:grad:ELBO:T} and the inclusion of control variate $\frac{1}{N}\sum_{i=1}^N \approxHstat{t}{i}$ for the computation of $\widehat{\grad} \ELBO{t}$}
\STATE \textit{// }$\grad \log \vd[\parvar_{t-1}]{t}(\sample{t}{i})$ \textit{is typically computed with automatic differentiation}
\end{algorithmic}
\end{algorithm}
\end{document}